\RequirePackage{fix-cm}
\documentclass[smallextended]{svjour3}       
\smartqed  
\usepackage{graphicx}

\usepackage{amsmath}
\usepackage{bbm}
\usepackage{graphicx}
\usepackage{amsfonts}

\begin{document}
\title{The relationships between message passing, pairwise, Kermack-McKendrick and stochastic SIR epidemic models \thanks{Research sponsered in part by The Leverhulme Trust Grant RPG-2014-341 to KJS}}

\titlerunning{Message passing, pairwise and Kermack-McKendrick epidemics}        

\author{Robert R. Wilkinson        \and
        Frank G. Ball  \and Kieran J. Sharkey
}


\institute{Robert R. Wilkinson \at
              Department of Mathematical Sciences, University of Liverpool, \\
               Peach Street,
              Liverpool L69 7ZL, UK \\
              \email{robertwi@liv.ac.uk}      \and Frank G. Ball
              \at
              School of Mathematical Sciences, University of Nottingham, \\University Park, Nottingham NG7 2RD, UK \\
              \email{Frank.Ball@nottingham.ac.uk}  
     \and Kieran J. Sharkey 
              \at
              Department of Mathematical Sciences, University of Liverpool, \\ Peach Street,
              Liverpool L69 7ZL, UK \\
              \email{kjs@liv.ac.uk}       \\
}

\date{Received: date / Accepted: date}

\maketitle

\begin{abstract}
	We consider a very general stochastic model for an SIR epidemic on a network which allows an individual's infectious period, and the time it takes to contact each of its neighbours after becoming infected, to be correlated. We write down the message passing system of equations
	for this model and prove, for the first time, that it has a unique feasible solution. We also generalise an earlier result by proving that this solution provides a rigorous upper bound for the expected epidemic size (cumulative number of infection events) at any fixed time $t > 0$.
	
	We specialise these results to a homogeneous special case where the graph (network) is symmetric. The message passing system here reduces to just four equations.  We prove that cycles in the network inhibit the
	spread of infection, and derive important epidemiological results
	concerning the final epidemic size and threshold behaviour for a major outbreak. For Poisson contact processes, this message passing system is equivalent to a non-Markovian
pair approximation model, which we show has well-known pairwise models as
	special cases.
	
	We show further that a sequence of message passing systems, starting with the homogeneous one just described, converges to the deterministic Kermack-McKendrick equations for this stochastic model. For Poisson contact and recovery, we show that this convergence is monotone, from which it follows that the message passing system (and
	hence also the pairwise model) here provides a better approximation to the
	expected epidemic size at time $t>0$ than the Kermack-McKendrick model.
	
\end{abstract}

\section{Introduction}

One of the earliest and most comprehensively analysed epidemic models is the susceptible-infected-recovered (SIR) model of Kermack and McKendrick (1927). In addition to providing insights into threshold behaviour and vaccination, it has also underpinned much subsequent work in applied mathematical epidemiology (Anderson and May 1992). A stochastic version, constructed from similar assumptions, was defined and analysed later (for example, Bailey 1975, Chapter 6) and it became of interest to understand the relationship between the two (Kurtz 1970, 1971; Barbour 1972, 1974).

More recently, various heterogeneities have been added to both deterministic and stochastic epidemic models. A particularly important one is the contact network which allows for specific relationships between pairs of individuals; see Danon et al. (2011) and Pastor-Satorras et al. (2015) for reviews. While it is straightforward to simulate stochastic epidemics on networks, deterministic approximations have also been developed to assist our understanding. Important examples of these include pair approximation (Keeling 1999; Sharkey 2008), message passing (Karrer and Newman 2010) and edge-based models (Miller et al. 2011).

The message passing approximation for stochastic epidemics was developed by Karrer and Newman (2010) and is central to the work that we present here. This approach allows one to exactly capture the marginal distributions for the health statuses of individuals (i.e whether they are susceptible, infected or recovered) when the contact network is a tree and provides useful rigorous bounds for these distributions otherwise. Notably, the message passing approach is also applicable to extremely general non-Markovian stochastic epidemics and the number of equations it requires scales linearly with the number of connected pairs of individuals; far fewer than the number of Kolmogorov forward equations for the Markovian case which scales exponentially with population size. Wilkinson and Sharkey (2014) showed that, when contact processes are assumed to be Poisson, a generalised version of the message passing equations is equivalent to a pairwise model that is defined at the level of individuals, thus unifying two major representations of epidemic dynamics. Their argument relies on the application of Leibniz's integral rule, so here we take the opportunity to provide sufficient conditions for the applicability of that rule in this context (appendix~\ref{cont}).

In section~\ref{second} we define a more general stochastic model which allows for realistic correlations between contact times and infectious periods. Specifically, it allows all of an individual's post-infection contact times (to each of its neighbours), and the negative of its infectious period, to be positively correlated. This could capture, for example, a scenario where infected individuals adopt some disease-combating behaviour such as taking antiviral medication, increasing the infectious contact times to all of their neighbours and decreasing their infectious period. We write down the message passing system for this stochastic model in subsection \ref{graphmes} and then, for the first time, provide a non-restrictive sufficient condition for the message passing equations to have a unique feasible solution (Theorem~\ref{th1}). This is important because so far, the message passing construction of Karrer and Newman has not been shown to give rise to a unique epidemic. We then, in subsection~\ref{boundsec}, extend the results found in Karrer and Newman (2010) and Wilkinson and Sharkey (2014) to this more general stochastic model; for example, the message passing system cannot underestimate the expected epidemic size at any time $t>0$, i.e. the expected number of
susceptibles infected during $(0,t]$ (Theorem~\ref{ineqtheorem} and Corollary~\ref{boundcor}). This is what led Karrer and Newman to describe the message passing system as providing a `worst case scenario'. 

For all of section \ref{symsec}, we focus on a special case of the above stochastic model which assumes a contact structure with a large amount of symmetry and that all individuals behave in the same way. We refer to this special case as the `homogeneous stochastic model'. The corresponding message passing system is written down in subsection~\ref{messagesym} and, after exploiting symmetries, this reduces to a system comprising of only four equations which we refer to as the `homogeneous message passing system'. This system is identical in form to a special case of the equations formulated by Karrer and Newman (2010, equations 26 and 27), although here it is related to a different stochastic model. We then obtain several epidemiologically relevant results in subsection~\ref{epiresults}: the stochastic epidemic is shown to be inhibited by cycles in the contact network (Theorem \ref{cyc}), a simple relation for an upper bound on the final epidemic size in the stochastic model is proved and sufficient conditions for no major outbreak in the stochastic model are found (Theorem \ref{outbreakthm}). The latter gives an upper bound on the critical vaccination coverage to prevent a major outbreak, assuming a perfect vaccination.    

As a special case of the general correspondence shown in Wilkinson and Sharkey (2014), the homogeneous message passing system has an equivalent non-Markovian pairwise model when the contact processes are Poisson. In subsection~\ref{pairwise} we write down these equations explicitly (Theorem~\ref{pairthm}). This pairwise model provides exactly the same epidemic time course as the homogeneous message passing system and hence exactly the same upper bound on the epidemic size at time $t$ (Corollary~\ref{paircor1}), and gives the same final epidemic size (Corollary~\ref{paircor2}).  Pairwise models are known to give good approximations of stochastic epidemic dynamics on networks in a broad range of cases (see, for example, Keeling (1999) and Sharkey (2008)). Thus the proof of equivalence when contact processes are Poisson suggests that message passing provides a good approximation as well as useful bounds. 

In subsection~\ref{large}, we derive the classic Kermack-McKendrick epidemic model as an asymptotic special case of the homogeneous message passing system (Theorem \ref{kermthm}). Notably, our derivation of such `deterministic' epidemic models from the homogeneous message passing system allows us to relate them explicitly to the stochastic model (see also, for example, Trapman (2007) and Barbour and Reinert (2013)). Thus, we are able to show that in the case where contact and recovery processes are independent and Poisson, the Kermack-McKendrick model bounds the expected epidemic size at time $t$ in the homogeneous stochastic model (Corollary~\ref{kermboundcor}). However, the bound is coarser than that provided by the homogeneous message passing system and the pairwise system, which therefore give a better approximation than the Kermack-McKendrick model. The paper ends with a brief discussion in section~\ref{disc}.  

      \section{The stochastic model (non-Markovin network-based SIR dynamics)}
      \label{second}
      
            We define a very general class of network-based stochastic epidemics which allow heterogeneous and non-Poisson individual-level processes, and heterogeneity in the initial states of individuals (including the case where the initial states of all individuals are non-random). 

Let $G=(\mathcal{V},\mathcal{E})$ be an arbitrary (possibly countably infinite) simple, undirected graph, where $\mathcal{V}$ is the set of vertices (individuals) and $\mathcal{E}$ is the set of undirected edges between vertices (throughout the paper we will use the terms `graph', `network' and `contact network' interchangeably). For $i \in \mathcal{V}$, let $\mathcal{N}_i=\{ j \in \mathcal{V} : (i,j) \in \mathcal{E} \}$ be the set of neighbours of $i$ and let $|\mathcal{N}_i|<\infty$. We assume that two individuals are neighbours if and only if at least one can make direct contacts to the other. A particular realisation of the stochastic model is specified as follows.  Each individual/vertex $i \in \mathcal{V}$ is assigned a set of numbers $\mathcal{X}_i$ relevant to the behaviour of $i$ and the spread of the epidemic:
       $$\mathcal{X}_i = \{ Y_i, \mu_i, \omega_{ji}\,(j \in \mathcal{N}_i) \} ,$$
       where $Y_i$ is equal to 1, 2, or 3, according to whether $i$ is instantaneously infected at $t=0$, initially susceptible or initially recovered/vaccinated, these being mutually exclusive; $\mu_i \in [0,\infty ]$ is $i$'s infectious period if $i$ is ever infected; $\omega_{ji} \in [0, \infty]$ is the time elapsing between $i$ first becoming infected and it making a contact to $j$, if $i$ is ever infected. Therefore, for $t \ge 0$, $i$ makes an infectious contact to $j$ at time $t$ if and only if (i) $i$ becomes infected at some time $s \le t$, (ii) $\omega_{ji}=t-s$, and (iii) $\omega_{ji}< \mu_i$. Susceptible individuals become infected as soon as they receive an infectious contact, and infected individuals immediately become recovered when their infectious period terminates (initially recovered/vaccinated individuals never become infected).  We let $\mathcal{X}= \cup_{i \in \mathcal{V}} \mathcal{X}_i$. Thus, the state of the population at time $t \in [0, \infty)$, which takes values in $\{S,I,R\}^{\mathcal{V}}$, is a function of $\mathcal{X}$.

The situation which we wish to consider is where $\mathcal{X}$ is a set of random variables, so from now on we refer to $Y_i, \mu_i, \omega_{ji}$, where $i \in \mathcal{V}, j \in \mathcal{N}_i$, as random variables. We use $r_i$ and $h_{ij}$ to denote the (marginal) probability density functions (PDFs) for $\mu_i$ and $\omega_{ij}$ respectively, and $z_i$ and $y_i$ to denote ${\rm P}(Y_i=2)$ and ${\rm P}(Y_i=3)$ respectively. Thus, ${\rm P}(Y_i=1)=1-y_i-z_i$. The probability that individual $i \in \mathcal{V}$ is in state $Z \in \{S,I,R\}$ at time $t \ge 0$ is denoted by $P_{Z_i}(t)$.

Importantly we assume that for every $i \in \mathcal{V}$,
              \begin{equation}
              	 \nonumber \mathcal{X}_i^*= \{ - \mu_i, \omega_{ji}\,(j \in \mathcal{N}_i) \}\end{equation} is a set of associated random variables, as defined by Esary et al. (1967) and discussed in this context by Donnelly (1993) and Ball et al. (2015). Additionally, we assume that the set of multivariate random variables $\{ \mathcal{X}_i: i \in \mathcal{V}\}$ is mutually independent, and that $Y_i$ and $\mathcal{X}_i^*$ are independent for all $i \in \mathcal{V}$. 
              	 A finite set of random variables, $T_1, T_2,\dots,T_n$ say, is associated (or positively correlated) if
              	 \begin{equation}
              	 \label{assocrv}
              	 {\rm E}[f(T_1, T_2,\dots,T_n)g(T_1, T_2,\dots,T_n)] \ge  {\rm E}[f(T_1, T_2,\dots,T_n)] {\rm E}[g(T_1, T_2,\dots,T_n)]
              	 \end{equation}
              	 for all non-decreasing real-valued functions $f, g$ for which the expectations in~\eqref{assocrv} exist.  Note that~\eqref{assocrv} implies that the correlation of any pair of these random variables is positive (i.e. $\ge 0$).  Further, if $T_1, T_2,\dots,T_n$ are mutually independent, then they are associated; see  Esary et al. (1967, Theorem 2.1).

 The above assumptions of association and independence are made so as to obtain the maximum amount of generality while the message passing and pairwise systems, which we shall define, give rigorous bounds on the expected dynamics in the stochastic model, and exact correspondence when the graph is a tree or forest.

              Our stochastic model represents a generalisation of that considered by Karrer and Newman (2010), and also generalises the model considered by Wilkinson and Sharkey (2014), which assumed that all of the elements of $\mathcal{X}$ are mutually independent. Here, we do not make this last assumption and allow all of an individual's post-infection contact times (to each of its neighbours), and the negative of its infectious period, to be positively correlated. This could capture, for example, the scenario where infected individuals tend to adopt some disease-combating behaviour, increasing the contact times to all of their neighbours and decreasing their infectious period.

 The model considered by Wilkinson and Sharkey (2014), which incorporates a directed graph, is equivalent to a special case of the above model. Directedness is still captured by the above model since, for any given $i \in \mathcal V$ and $ j \in \mathcal{N}_i$, $\omega_{ij}$ and $\omega_{ji}$ are assigned independently.

\subsection{\bf{The message passing system and its unique solution}}
\label{graphmes}

Following Wilkinson and Sharkey (2014), we apply the message passing approach of Karrer and Newman (2010) to the stochastic model defined in section \ref{second}. Recall that message passing relies on the concept of the cavity state in order to simplify calculations. An individual is placed into the cavity state by cancelling its ability to make contacts. This does not affect its own fate but it does affect the fates of others because it cannot pass on the infection.

For arbitrary $i \in \mathcal{V}$ and neighbour $j \in \mathcal{N}_i$, let $H^{i \leftarrow j}(t)$ denote the probability that $i$, when in the cavity state, does not receive an infectious contact from $j$ by time $t$. We can now write:
\begin{equation} \label{firstH}
H^{i \leftarrow j}(t)=1- \int_0^t f_{ij}(\tau) \big(1-y_j-z_j \Phi_i^j(t-\tau)  \big) \mbox{d} \tau,
\end{equation}
where $f_{ij}(\tau) \Delta \tau = h_{ij}(\tau) P(\mu_j > \tau \mid \omega_{ij}= \tau) \Delta \tau$ is the probability ($+o(\Delta \tau)$) that $j$ makes an infectious contact to $i$ during the time interval $[\tau,\tau + \Delta \tau)$ (for $\Delta \tau \to 0$), where time $\tau$ is measured from the moment $j$ becomes infected, and $\Phi_i^j(t)$ is the probability that $j$ does not receive any infectious contacts by time $t$ when $i$ and $j$ are both in the cavity state. Note that although the stochastic model considered here is more general, $H^{i \leftarrow j}(t)$ may still be expressed, as in \eqref{firstH}, similarly to equation 1 in Wilkinson and Sharkey (2014), because $\{\mathcal{X}_i: i \in \mathcal{V}\}$ is mutually independent and $Y_i$ is independent from $\mathcal{X}_i^*$ for all $i \in \mathcal{V}$. 

To obtain a solvable system, the probability $H^{i \leftarrow j}(t)$ is approximated by $F^{i \leftarrow j}(t)$, where $F^{i \leftarrow j}(t) \, \, (i \in \mathcal{V}, j \in \mathcal{N}_i)$ satisfies
\begin{equation} \label{Fgen1}
F^{i \leftarrow j}(t) = 1- \int_0^t f_{ij} (\tau) \Big(1-y_j-z_j \prod_{k \in \mathcal{N}_j \setminus i} F^{j \leftarrow k}(t-\tau)  \Big) \mbox{d} \tau.
\end{equation}
Any solution of (\ref{Fgen1}) which gives $F^{i \leftarrow j}(t) \in [0,1]$ for all $t \ge 0$, and all $i \in \mathcal{V} , j \in \mathcal{N}_i$, is called feasible. It was shown by Wilkinson and Sharkey (2014), following Karrer and Newman (2010), that a feasible solution exists as the limit of an iterative procedure.

The message passing system can now be defined (for $ i \in \mathcal{V}$):
\begin{eqnarray} \label{hom1}
S_{\text{mes}}^{(i)}(t) & = & z_i \prod_{j \in \mathcal{N}_i} F^{i \leftarrow j}(t), \\
I_{\text{mes}}^{(i)}(t) & =& 1 - S_{\text{mes}}^{(i)}(t) - R_{\text{mes}}^{(i)}(t), \\ \label{Rm}
R_{\text{mes}}^{(i)}(t) & = & y_i + \int_0^t r_i(\tau)[1-y_i-S_{\text{mes}}^{(i)}(t-\tau)] \mbox{d} \tau,
\end{eqnarray}
where the variables on the left-hand side approximate $P_{S_i}(t)$, $P_{I_i}(t)$ and $P_{R_i}(t)$ respectively (recall that $P_{S_i}(t)$, $P_{I_i}(t)$ and $P_{R_i}(t)$ are respectively the probability that individual $i$ is susceptible, infective and recovered-or-vaccinated at time $t$). Numerical evidence for the effectiveness of the message passing system, in capturing the expected dynamics of the stochastic model, can be seen in Figures 1 and 2 of Wilkinson and Sharkey (2014).

Note that the dimension of the message passing system \eqref{Fgen1}-\eqref{Rm} is appreciably smaller than that of the Kolomogorov forward equations for the case where the dynamics are Markovian. Suppose that $|\mathcal{V}|=N$. Then the forward equations have dimension $3^N$ and the message passing system has dimension at most $N(N-1)+3N$. In many cases, symmetries can be exploited to reduce the dimension of both the forward equations, see e.g. Simon et al. (2011), and the message passing system. However, the message passing system is still typically much smaller and can be very small, as in the model studied in section~\ref{symsec}. 

\begin{theorem}[Uniqueness of the feasible solution of the message passing system]
	\label{th1}
	Assume that
	$$ \sup_{i \in \mathcal{V}} |\mathcal{N}_i| < \infty \quad \mbox{and} \quad \sup_{(i,j) \in \mathcal{E}} \left(  \sup_{\tau \ge 0}f_{ij}(\tau) \right)  < \infty . $$
	Then there is a unique feasible solution of equations \rm(\ref{Fgen1})\it-\rm(\ref{Rm})\it and the feasible $F^{i \leftarrow j}(t)$ are continuous and non-increasing for all $i \in \mathcal{V}, j \in \mathcal{N}_i$.
\end{theorem}
\begin{proof}
	See appendix~\ref{feasunique}. \qed
\end{proof}

It was shown by Wilkinson and Sharkey (2014) that when the graph is finite and $f_{ij}(\tau)=T_{ij} {\rm e}^{- T_{ij} \tau} \int_{\tau}^{\infty} r_j(\tau') \mbox{d} \tau' \, \, (i \in \mathcal{V}, j \in \mathcal{N}_i)$, where $T_{ij} \in (0, \infty)$, i.e. contact processes are Poisson and independent of recovery processes, then the message passing system \eqref{Fgen1}-\eqref{Rm} is equivalent to an individual-level pairwise system of integro-differential equations. It now follows that this pairwise system of equations also has a unique feasible solution.

The message passing system \eqref{Fgen1}-\eqref{Rm}, which coincides with that given in Wilkinson and Sharkey (2014) although the underlying model here is more general, differs from the message passing system in Karrer and Newman (2010) in that the probability an individual is initially infected need not be the same for all individuals, and individuals may be initially recovered or vaccinated. The system \eqref{Fgen1}-\eqref{Rm} also accounts for heterogeneity in the recovery and contact processes. A key use of message passing equations is that they yield a rigorous upper bound for the mean spread in the underlying stochastic epidemic. In the next subsection, we show that this property extends to our more general  model. 
	
	\subsection{\bf{Bounding the expected epidemic size at time $\boldsymbol{t}$}}
	\label{boundsec}
	
For $t \ge 0$, let $X(t)$ denote the number of susceptibles at time $t$. Thus, $X(0)-X(t)$ is the total number of individuals infected by time $t$ not counting those infected at $t=0$. We refer to this quantity as the epidemic size at time $t$. 

	\begin{theorem}[Message passing bounds the marginal distribution for the health status of an individual] \label{ineqtheorem} \newline
For all $t \ge 0$ and all $i \in \mathcal{V}$,
  \begin{eqnarray}  \label{ineq1}
	P_{S_i}(t) &\ge & S_{\text{mes}}^{(i)}(t), \\ \label{ineq2}
	P_{R_i}(t) &\le & R_{\text{mes}}^{(i)}(t) ,
	\end{eqnarray}
  with equality if $G$ is a tree or forest.
\label{boundthm}
	\end{theorem}
	\begin{proof}
	In the case where $\mathcal{X}$ is mutually independent and $\mathcal{V}$ is finite, this is proved by Wilkinson and Sharkey (2014) and Ball et al. (2015) by generalising Karrer and Newman (2010). The proof for our current more general model is in appendix~\ref{boundappend}.\qed
\end{proof}

For $t \ge 0$, let $Z(t)$ denote the number of recovered-or-vaccinated individuals at time $t$.  The following corollary follows immediately from Theorem~\ref{ineqtheorem} on noting that, for $t \ge 0$, 
\begin{equation*}
{\rm E}[X(t)]=\sum_i P_{S_i}(t) \qquad \mbox{and} \qquad {\rm E}[Z(t)]=\sum_i P_{R_i}(t).
\end{equation*} 

\begin{corollary}
	\label{boundcor}
For all $t \ge 0$, we have ${\rm E}[X(t)] \ge \sum_{i}S^{(i)}_{\text{mes}}(t)$ and ${\rm E}[Z(t)] \le \sum_{i}R^{(i)}_{\text{mes}}(t)$, with equality occurring when the graph is a tree or forest. The expected epidemic size at time $t$ is given by $E[X(0)-X(t)]=\sum_{i \in \mathcal{V}}z_i - E[X(t)]$. Thus, since we have a lower bound on $E[X(t)]$ we also have an upper bound on the expected epidemic size at time $t$.
\end{corollary}

\section{The homogeneous stochastic model}

\label{symsec}
In this section we consider a special case of the stochastic model, and we refer to this special case as `the homogeneous stochastic model'. In the homogeneous stochastic model, the graph is symmetric and connected. Examples of symmetric connected graphs include complete graphs, ring lattices, infinite square lattices and Bethe lattices. In a symmetric graph, each individual has the same (finite) number $n$ of neighbours, and we say that the graph is $n$-regular. To avoid triviality we assume $n \ge 2$.
\begin{definition}
 A graph $G=(\mathcal{V},\mathcal{E})$ is called symmetric if it is arc-transitive; i.e. for any two ordered pairs of neighbours $i,j$, and $i',j'$, there exists a graph-automorphism which maps $i$ to $i'$ and $j$ to $j'$ (Godsil and Royle 2001).\label{defsym}\end{definition}

Additionally, in the homogeneous stochastic model, the joint distribution of $(Y_i, \mu_i , \omega_{j i}(j \in \mathcal{N}_i))$ is symmetric in its last $n$ arguments and is the same for all $i \in \mathcal{V}$. Thus, it is impossible to distinguish between any two individuals by their behaviour or by their position in the graph. Note that we have not precluded the variables in $\mathcal{X}_i^*$ from being non-trivially associated (for all $i \in \mathcal{V}$), i.e. $i$'s infectious period and the time it takes for it to contact each of its neighbours, after infection, may all be non-trivially correlated.

 We use $r$ and $h$ to denote the (marginal) PDFs for $\mu_i$ and $\omega_{ij}$ respectively, and $z$ and $y$ to denote ${\rm P}(Y_i=2)$ and ${\rm P}(Y_i=3)$ respectively. Thus, ${\rm P}(Y_i=1)=1-y-z$. To avoid triviality, we assume that $0<z<1$ and $0 \le y<1-z$.

  Owing to symmetry (in this special case), the probability distribution for the health status of an individual is the same for all individuals, i.e. for all $i,i' \in \mathcal{V}$ and all $t \ge 0$, we have $P_{S_i}(t)=P_{S_{i'}}(t), P_{I_i}(t)=P_{I_{i'}}(t)$ and $P_{R_i}(t)=P_{R_{i'}}(t)$ (let $P_S(t), P_I(t)$ and $P_R(t)$ denote these quantities). Similarly, for all $i \in \mathcal{V}, j \in \mathcal{N}_i$ and all $i' \in \mathcal{V}, j' \in \mathcal{N}_{i'}$, and all $t \ge 0$, we have $H^{i \leftarrow j}(t)= H^{i' \leftarrow j'}(t)$ (let $H_{\text{sym}}(t)$ denote this quantity).

\subsection{\bf{The homogeneous message passing system}}
\label{messagesym}
 For the homogeneous stochastic model, (\ref{Fgen1}) becomes
\begin{equation} \label{otherF}
F^{i \leftarrow j}(t) = 1- \int_0^t f(\tau) \Big(1-y-z \prod_{k \in \mathcal{N}_j \setminus i} F^{j \leftarrow k}(t-\tau)  \Big) \mbox{d} \tau \quad (i \in \mathcal{V}, j \in \mathcal{N}_i),
\end{equation}
where we have used $f_{ij}(\tau)=f_{i'j'}(\tau)$ for all $i \in \mathcal{V}, j \in \mathcal{N}_i$, and all $i' \in \mathcal{V}, j \in \mathcal{N}_{i'}$, and all $\tau \ge 0$, and we let $f(\tau)$ denote this quantity.

The arc-transitivity of symmetric graphs and the symmetry in (\ref{otherF}) allow us to simplify (\ref{Fgen1})-(\ref{Rm}), and to write down the full homogeneous message passing system as:
       \begin{eqnarray}
       \label{hom3}
       S_{\text{mes}}(t)&=&z F_{\text{sym}}(t)^{n}, \\ \label{Ipair}
       I_{\text{mes}}(t)&=& 1-S_{\text{mes}}(t)-R_{\text{mes}}(t), \\ \label{Rpair}
       R_{\text{mes}}(t)&=& y + \int_0^t r(\tau)[1-y-S_{\text{mes}}(t-\tau)] \mbox{d} \tau,
       \label{hom}
       \end{eqnarray}
       where
                            \begin{eqnarray} \label{Fgraph}
                              F_{\text{sym}}(t)& =& 1 - \int_0^t f(\tau) \left[ 1 - y -   zF_{\text{sym}}(t-\tau)^{n-1}    \right] \mbox{d} \tau.
                            \end{eqnarray}
                             In deriving these equations, we have used
                             $F^{i \leftarrow j}(t)= F^{i' \leftarrow j'}(t)$ for all $i \in \mathcal{V}, j \in \mathcal{N}_i$ and all $i' \in \mathcal{V}, j' \in \mathcal{N}_{i'}$, and all $t \ge 0$, and we let $ F_{\text{sym}}(t)$ denote this quantity. Note that we have also made use of the fact that every individual has $n$ neighbours. This system is identical in form (when vaccination is disallowed) to the message passing system for the configuration network model provided by Karrer and Newman (2010, equations 26 and 27, making use of equations 1, 4 and 5), in the case where every individual has $n$ neighbours with probability 1. From Theorem~\ref{th1}, we know that if $\sup_{\tau \ge 0} f(\tau) < \infty$ then (\ref{Fgraph}) has a unique feasible solution.
                             
 For clarity we write out these equations for the simplifying cases of Poisson transmission and recovery processes, and Poisson transmission and fixed (non-random) recovery.                         \newline 
                             
                              \noindent {\bf Example 1: Poisson transmission and recovery}
                             
                                                                             For independent Poisson transmission and recovery processes (specifically, $\tau_i$ and $\omega_{ji}$ are independent and exponentially distributed with rates $\gamma$ and $\beta$ respectively), with $f(\tau)=\beta {\rm e}^{-(\beta + \gamma) \tau}$, the homogeneous message passing system can be solved via the following ordinary differential equations (ODEs):
                                                                             \begin{eqnarray} \label{case1a}
                                                                             \dot{F}_{\text{sym}}(t)&=& \gamma \Big(1-F_{\text{sym}}(t) \Big) - \beta \Big(F_{\text{sym}}(t)  - y - z F_{\text{sym}}(t)^{n-1}  \Big), \\ \label{case1b}
                                                                             \dot{R}_{\text{mes}}(t) &=& \gamma I_{\text{mes}}(t),
                                                                             \end{eqnarray}
                                                                            with $S_{\text{mes}}(t)$ and $I_{\text{mes}}(t)$ given by (\ref{hom3}) and (\ref{Ipair}). \newline

                              \noindent {\bf Example 2: Poisson transmission and fixed recovery}
                             
                              For Poisson transmission processes and a fixed recovery period (specifically, $\tau_i$ is non-random with value $R \in [0, \infty]$ and $\omega_{ji}$ is exponentially distributed with rate $\beta$), with $f(\tau)=\beta {\rm e}^{- \beta \tau}(1- \theta(t-R))$ where $\theta$ is the Heaviside step function, the homogeneous message passing system can be solved using the following delay differential equation:
                              \begin{eqnarray} \nonumber \label{case2a}
                              \dot{F}_{\text{sym}}(t)&=& - \beta \Big(F_{\text{sym}}(t) - y - zF_{\text{sym}}(t)^{n-1}  \\
                              &&- \theta(t-R) {\rm e}^{- \beta R} \big( 1 - y - z F_{\text{sym}}(t-R)^{n-1}  \big) \Big) ,  \end{eqnarray}
                             with
                             \begin{eqnarray}  \label{case2b}
                              R_{\text{mes}}(t) &=& y + \theta(t-R) \big( 1 - y - S_{\text{mes}}(t-R) \big),
                              \end{eqnarray}
                              and $S_{\text{mes}}(t)$ and $I_{\text{mes}}(t)$ given by (\ref{hom3}) and (\ref{Ipair}). \newline
                             
                             Other choices of $f(\tau)$ exist which allow the message passing system to be solved via (non-integro) differential equations, such as the top hat function (Karrer and Newman 2010, equation 33).

                             \subsection{\bf{Epidemiological results}}
                             \label{epiresults}
                        
                        As well as bounding/approximating (or correctly computing in the case of an infinite regular tree) the expected fractional epidemic size at time $t \ge 0$, the homogeneous message passing system generates other epidemiologically relevant results for the stochastic model, as demonstrated here.
                        
                        \begin{theorem}[Cycles in the network inhibit the stochastic epidemic] \label{cyc}
                        	\label{theorem3}
                        	Suppose that $\sup_{\tau \ge 0}f(\tau) < \infty$. The probability of an arbitrary individual being susceptible at a given time, for the $n$-regular Bethe lattice \rm(\it infinite tree\rm)\it, is less than or equal to this quantity for all other $n$-regular symmetric graphs \rm(\it where the homogeneous stochastic model is otherwise unchanged\rm)\it. The same holds for the probability of an arbitrary individual being recovered except with the inequality reversed.
                        \end{theorem}
                        
                        \begin{proof}
                        	From Theorem~\ref{ineqtheorem}, we know that system (\ref{hom3})-(\ref{Fgraph}) cannot overestimate the probability of an arbitrary individual being susceptible at time $t$ and cannot underestimate the probability of an arbitrary individual being recovered at time $t$. However, also from Theorem \ref{ineqtheorem}, the system is exact if the graph is a tree.  \qed
                        \end{proof}

                        Theorem \ref{theorem3} suggests that, all other things being equal, an infection will have the greatest impact by time $t$ when the contact structure is most tree-like. Indeed, it is known that clustering and the presence of cycles in the graph may slow down and limit the spread of an infection (see Miller (2009) and references therein).

\begin{theorem}[Final epidemic size relation and sufficient conditions for no major outbreak]
	For all $t \ge 0$,
	\label{outbreakthm}
	\begin{equation} \label{finsizebound}
	S_{\text{mes}}(\infty) \le  P_{S}(t), \quad  R_{\text{mes}}(\infty) \ge  P_{R}(t),
	\end{equation}
	where $S_{\text{mes}}(\infty) \equiv \lim_{t \to \infty}S_{\text{mes}}(t)$ may be computed as the unique solution in $[0,z]$ of
	\begin{equation}  \left( \frac{ S_{\text{mes}}(\infty)}{z} \right)^{\frac{1}{n}}=1-p + py+ pz \left(\frac{ S_{\text{mes}}(\infty)}{z} \right)^{\frac{n-1}{n}} ,\label{finS} \end{equation} 
 with $p \equiv  \int_0^{\infty} f(\tau) \mbox{d} \tau$, and $R_{\text{mes}}(\infty)=1-S_{\text{mes}}(\infty)$.
 
 Further, when the fraction initially infected is small, i.e. $z \to 1 - y$ from below, then
 	\begin{equation} \label{noepidemic} P_S(\infty) = P_S(0) \qquad \mbox{if} \qquad y \ge 1 - \frac{1}{R_0} \quad \text{or} \quad R_0 \le 1, \end{equation}
 	where $R_0 \equiv (n-1)p$.	\rm(\it This means that if each individual is independently vaccinated with probability greater than or equal to $1- 1/ R_0$, or if $R_0 \le 1$, then a major outbreak of the disease is impossible.\rm)\it
\end{theorem}
\begin{proof}
	Equation \eqref{finsizebound} follows from Theorem \ref{boundthm} and the observation that $P_S(t)$ and $P_R(t)$ are non-increasing and non-decreasing respectively. 
	
	The feasible $F_{\text{sym}}(t)$ is non-increasing (see Theorem~\ref{th1}), so it converges to some $F_{\text{sym}}(\infty) \in [0,1]$ as $t \to \infty$. Note also that, by definition, $\int_0^t f(\tau) \mbox{d} \tau$ converges to $p \in [0,1]$ as $t \to \infty$. Now, using (\ref{Fgraph}), we can write $F_{\text{sym}}(t)=1-\int_0^{\infty} f_t(\tau) \mbox{d} \tau$, where $f_t(\tau)=f(\tau)(1-y-z F_{\text{sym}}(t-\tau)^{n-1})$ for $\tau \in [0,t]$ and is equal to zero for $\tau>t$. Note that $f_t(\tau)$ converges pointwise to $f(\tau)(1-y-z F_{\text{sym}}(\infty)^{n-1})$ as $t \to \infty$. Thus, since $0 \le f_t(\tau) \le f(\tau)$ for all $t, \tau \ge 0$, we can use the dominated convergence theorem to obtain, c.f. Karrer and Newman (2010, equations 23 and 24),
	\begin{equation} \label{finF} F_{\text{sym}}(\infty)=1-\int_0^{\infty} \lim_{t \to \infty} f_t(\tau) \mbox{d} \tau=1-p \left( 1 - y -z F_{\text{sym}}(\infty)^{n-1} \right) . \end{equation}
	Taking the limit as $t \to \infty$ in (\ref{hom3}), and making use of (\ref{finF}), proves equation \eqref{finS}. It is straightforward to show by graphical means that \eqref{finS} has a unique solution in $[0,z]$. In the case where $z \to 1-y$ from below, it is also straightforward to show by graphical means that, after setting $z=1-y$ in \eqref{finS}, then $S_{\text{mes}}(\infty)=z(=S_{\text{mes}}(0)=P_S(0))$ is the only solution in $[0,z]$ if $y \ge 1 - 1/R_0$ ($R_0 \le 1$ implies this condition). Equation \eqref{noepidemic} is then proved by noting that $P_S(t) \ge S_{\text{mes}}(t)$, and $P_S(t)$ is non-increasing from $P_S(0)=S_{\text{mes}}(0)$. \qed
	\end{proof}
	
		Equation~\ref{finS} is consistent with the final size relation given by Diekmann et al. (1998) (equations 5.3 and 5.4) for a regular random graph in the limit of large population size.

\begin{remark}	Consider an infinite sequence of finite homogeneous stochastic models, indexed by $m$, where $y_m=y \in [0,1)$ for all $m$, and where $N_m \to \infty$, $p_m(n_m-1)\to R_0 < \infty$, $z_m \to 1-y$, as $m \to \infty$ (here, $N_m$ denotes the number of individuals in the $m$th model). This does not preclude the expected number of initial infectives from tending to some positive number, or even diverging, as $m \to \infty$. It is straightforward that, in the limit of this sequence, the sufficient conditions for no major outbreak in Theorem~\ref{outbreakthm} still hold. Note that if in addition we have $n_m \to \infty$ as $m \to \infty$, then the final size relation for the homogeneous message passing system (in this limit) becomes, using \eqref{finS} with $z=1-y$,
	$$ \frac{ S_{\text{mes}}(\infty)}{1-y} = {\rm e}^{-R_0(1-S_{\text{mes}}(\infty)-y)} .  $$
	This is a well-known final size relation in the mean field literature, although usually vaccination is not included (see Miller (2012) for a discussion of derivations of this relation).
	\end{remark}
	
\subsection{\bf{The homogeneous message passing system gives the same epidemic time course as a pairwise model}}
\label{pairwise}
Here we show that a generalised pairwise SIR model, with well-known pairwise models as special cases, gives the same epidemic time course as the homogeneous message passing system. This allows us to prove epidemiological results for the generalised pairwise model. Since pairwise models are known to give good approximations of stochastic epidemic dynamics on networks (see, for example, Keeling (1999) and Sharkey (2008)), this also strengthens the case for the message passing system being a good approximation. 
\begin{theorem}[Equivalence of the message passing and pairwise models]
	\label{pairthm}
For the homogeneous stochastic model, assume that the contact processes are Poisson with rate $\beta$ and that they are independent from the recovery processes, such that $f(\tau)=\beta {\rm e}^{- \beta \tau} \int_{\tau}^{\infty}r(\tau') \mbox{d} \tau'$. Assume also that $r(\tau)$ is continuous. Then,  
\begin{eqnarray} \label{mspair}
\dot{[S]}(t) &=& - \beta [SI](t), \\ \label{mipair}
\dot{[I]}(t) &=&  \beta [SI](t) - \int_0^t r(\tau) \beta [SI](t- \tau) {\rm d} \tau - r(t) N (1-y-z),  \\ \label{msspair}
\dot{[SS]}(t) & = & - 2 \beta \frac{n-1}{n} \frac{[SS](t) [SI](t)}{[S](t)}, \\  \label{msipair}
 \nonumber
\dot{[SI]}(t)&=& - \beta \left( \frac{n-1}{n} \right) \frac{[SI](t) [SI](t)}{[S](t)} \\ \nonumber
&& - \beta [SI](t) \\ \nonumber
&& + \beta \left( \frac{n-1}{n} \right) \frac{[SS](t)[SI](t)}{[S](t)} \\ \nonumber
&& - \int_0^t {\rm e}^{- \beta \tau} r(\tau) \beta \left( \frac{n-1}{n} \right) \frac{[SS](t- \tau)[SI](t - \tau)}{[S](t - \tau)}\\ \nonumber
&& \times \, \, {\rm exp}\Bigg( -\int_{t - \tau}^{t}
 \beta \left( \frac{n-1}{n} \right) \frac{[SI](\tau')}{[S](\tau')} {\rm d} \tau' \Bigg) {\rm d} \tau \\ \nonumber && -n N z {\rm e}^{- \beta t} r (t) (1-y-z) \, {\rm exp}\Bigg( -\int_0^{t}
  \beta \left( \frac{n-1}{n} \right) \frac{[SI](\tau)}{[S](\tau)} {\rm d} \tau \Bigg), \\
\end{eqnarray}
where
\begin{eqnarray} \label{pairwisevar}
\left[S \right](t) & \equiv & N S_{\text{mes}}(t), \\ \label{pairmesi}
\left[I \right](t) & \equiv & N I_{\text{mes}}(t), \\ \label{msspair2}
\left[SS \right](t) & \equiv & n N SS_{\text{mes}}(t) \equiv n N z^2 F_{\text{sym}}(t)^{2(n-1)}, \\ \label{msipair2}
\left[SI \right](t) & \equiv & n N SI_{\text{mes}}(t) \equiv n N  z F_{\text{sym}}(t)^{n-1} \Bigg( \frac{-  \dot{F}_{\text{sym}}(t)}{\beta} \Bigg),
\end{eqnarray}
and $N$ is a positive number. 
\end{theorem}
\begin{proof}
	see appendix \ref{pairappendix}.  \qed
	\end{proof}
			
\begin{corollary}[At all time points the pairwise model cannot underestimate the expected epidemic size]
\label{paircor1}
If, in the homogeneous stochastic model, contact processes are Poisson with rate $\beta$, i.e. the marginal distribution for $\omega_{ji}$ is exponential with parameter $\beta$ for all $i \in \mathcal{V}, j \in \mathcal{N}_i$, and these are independent from the infectious periods, then
\begin{equation} \nonumber
[S](t)/N \le  P_{S}(t), \quad  [R](t)/N \ge  P_{R}(t) \qquad (t \ge 0),
\end{equation}
where $[R](t)=N-[S](t)-[I](t)$.
\end{corollary}
\begin{proof}
This follows immediately from Theorems \ref{boundthm} and \ref{pairthm}. \qed
\end{proof}

\begin{corollary}[Final epidemic size equation for the pairwise model]
	\label{paircor2}
		\begin{equation} \nonumber  \left( \frac{ [S](\infty)}{Nz} \right)^{\frac{1}{n}}=1-p + py+ pz \left(\frac{ [S](\infty)}{Nz} \right)^{\frac{n-1}{n}} , \end{equation} 
		where $[S](\infty) \equiv \lim_{t \to \infty}[S](t)$ and $p \equiv  \int_0^{\infty} \beta {\rm e}^{- \beta \tau'} \int_{\tau'}^{\infty}r(\tau) \mbox{d} \tau \mbox{d} \tau'$.
	\end{corollary}
	\begin{proof}
		This follows immediately from Theorems \ref{outbreakthm} and \ref{pairthm}. \qed
		\end{proof}

	Note that~\eqref{mspair}-\eqref{msipair} constitute a closed system for the variables $[S](t), [I](t), \newline [SS](t)$ and $[SI](t)$ (if $[S](t)=0$ then the right-hand sides of \eqref{msspair} and \eqref{msipair} are undefined, but in this case the left-hand sides are equal to zero). With reference to \eqref{msspair2} and \eqref{msipair2}, the quantities $SS_{\text{mes}}(t)$ and $SI_{\text{mes}}(t)$ are constructed to capture/approximate, for any given pair of neighbours at time $t$, the probability that they are both susceptible and the probability that the first is susceptible while the second is infected respectively (see appendix~\ref{pairappendix}). The system \eqref{mspair}-\eqref{msipair} also follows directly from application of the individual-level pairwise equations in Wilkinson and Sharkey (2014, equations 8 and 9). In the case where the infectious period is exponentially distributed and letting $N$ be the population size, (\ref{mipair}) and (\ref{msipair}) simplify to ODEs, and the pairwise (without clustering) model of Keeling (1999) is obtained. Similarly, after substituting $r(\tau)=\delta (t-R)$, where $\delta$ is the Dirac delta function, into \eqref{mipair} and \eqref{msipair}, the pairwise model of Kiss et al. (2015) for a non-random infectious period of duration $R$ is obtained (except that the last term in (\ref{mipair}) and the last term in (\ref{msipair}), which relate to the behaviour of the initial infectives, need to be neglected). However, it may be more efficient to solve the simpler message passing systems (via (\ref{case1a})-(\ref{case1b}) and (\ref{case2a})-(\ref{case2b}) respectively) and then, if pairwise quantities are required, these can be computed using (\ref{msspair2}) and (\ref{msipair2}). 
	
	As part of the proof of equivalence between message passing and pairwise models that we present here, we also close a gap in the arguments of Wilkinson and Sharkey (2014) by demonstrating sufficient conditions for the valid application of Leibniz's integral rule (appendix~\ref{cont}) in the derivation of the pairwise equations from the message passing equations (appendix~\ref{pairappendix}). 

\subsection{\bf{The homogeneous message passing system gives the same epidemic time course as the Kermack-McKendrick model (asymptotically)}}
\label{large}

Here, we consider a sequence of homogeneous stochastic models where the regular degree $n$ tends to infinity. As $n \to \infty$, an individual is able to make contacts to a number of neighbours which tends to infinity, so to obtain a finite limit we assume that the infection function $f(\tau)$ depends on $n$ (which we write $f_n(\tau)$) such that:
\begin{equation}
\nonumber
\lim_{n \to \infty}n f_n(\tau)= f^*(\tau)< \infty \qquad(\tau \ge 0).
\end{equation}
Note that, in the limit of large $n$, transmission is frequency dependent and the expected number of infectious contacts
made by a given infected individual during the time interval $(t_1,t_2)$ is $\int_{t_1}^{t_2} f^*(\tau)\mbox{d} \tau$, where time is measured from the moment the individual first became infected.

The deterministic model proposed by Kermack and McKendrick (1927) is as follows:
\begin{eqnarray} \label{hom11}       \dot{S}(t)&=&  S(t) \left[    \int_0^t  f^*(\tau) \dot{S}(t- \tau) \mbox{d} \tau - I(0)f^*(t) \right],  \\        I(t)&=& 1-S(t)-R(t), \\
R(t)&=& R(0)+ \int_0^t r(\tau)[1-R(0)-S(t-\tau)] \mbox{d} \tau.
\label{hom2}
\end{eqnarray}
Equations 12-15 of Kermack and McKendrick (1927) may be obtained from~\eqref{hom11}-\eqref{hom2} after multiplying through by the total population size $N$ in their paper.

The following theorem shows that, under this limiting regime and mild further conditions, the homogeneous message passing system gives the same epidemic time course as the model of Kermack and McKendrick (1927).  For $n=1,2,\dots$, let $S_{\text{mes}(n)}(t),I_{\text{mes}(n)}(t)$ and  $R_{\text{mes}(n)}(t)$ denote the message passing system given by~\eqref{hom3}-\eqref{Fgraph}, where $F_{\text{sym}}(t)$ is replaced by $F_{\text{sym}(n)}(t)$, which satisfies~\eqref{Fgraph} with $f(\tau)$ replaced by $f_n(\tau)$.

\begin{theorem}[Deriving the Kermack-McKendrick model from message passing]
	Suppose that for all $T\ge0$,
	\begin{enumerate}
		\item[(i)] $\epsilon_n(T)=\sup_{0 \le t \le T}|nf_n(t)-f^*(t)| \to 0$ as $n \to \infty$, 
		\item[(ii)] $M_T=\sup_{0 \le t \le T} f^*(t) < \infty$,
	\end{enumerate}
	and that, for all $n=1,2,\dots,$ 
	\begin{enumerate}
		\item[(iii)] $f_n(t)$ is continuously differentiable,
		\item[(iv)] $\left(S_{\text{mes}(n)}(0),I_{\text{mes}(n)}(0),R_{\text{mes}(n)}(0)\right)=(S(0),I(0),R(0))=(z,1-z-y,y)$.
	\end{enumerate}
	Then, for all $T>0$,
	\begin{eqnarray}
	\label{theorem5S}
	\lim_{n \to \infty} \sup_{0 \le t \le T}\left|S_{\text{mes}(n)}(t)-S(t)\right|&=&0,\\
	\lim_{n \to \infty} \sup_{0 \le t \le T}\left|I_{\text{mes}(n)}(t)-I(t)\right|&=&0, \label{theorem5I}\\
	\lim_{n \to \infty} \sup_{0 \le t \le T}\left|R_{\text{mes}(n)}(t)-R(t)\right|&=&0. \label{theorem5R}
	\end{eqnarray}
	\label{kermthm}
\end{theorem}

\begin{proof}
	Fix $T>0$ and note first from~\eqref{Fgraph} that, for feasible $F_{\text{sym}(n)}(t)$ and all $t \in [0,T]$,
	\begin{equation*} 
	1 \ge F_{\text{sym}(n)}(t) \ge 1- \int_0^t f_n(\tau) \mbox{d}\tau \quad (n=1,2,\dots).
	\end{equation*}
	Now $n\int_0^t f_n(\tau) \mbox{d}\tau \le T(M_T+\epsilon_n(T))$, for all $t \in [0,T]$, so conditions (i) and (ii) imply that
	there exists $\epsilon_n^{(1)}(T)\ge 0$ such that for all $t \in [0,T]$,
	\begin{equation} \label{Flargen}
	1 \ge F_{\text{sym}(n)}(t) \ge 1-\epsilon_n^{(1)}(T) \quad (n=1,2,\dots),
	\end{equation}
	where $\epsilon_n^{(1)}(T) \to 0$ as $n \to \infty$.  Thus, for all sufficiently large $n$,  $ F_{\text{sym}(n)}(t)$ is non-zero for all $t \in [0,T]$.  
	
	Differentiating~\eqref{hom3} yields
	\begin{eqnarray} 
	\dot{S}_{\text{mes}(n)}(t)&=&nz F_{\text{sym}(n)}(t)^{n-1}\dot{F}_{\text{sym}(n)}(t), \label{Smean1}
	\end{eqnarray}
	and differentiating~\eqref{Fgraph}, using Leibniz's integral rule (see appendix \ref{cont}), gives
	\begin{eqnarray} \nonumber
	\dot{F}_{\text{sym}(n)}(t)&=&-f_n(t)(1-y-z) \\
	&&+(n-1)z\int_0^t f_n(\tau) F_{\text{sym}(n)}(t- \tau)^{n-2} \dot{F}_{\text{sym}(n)}(t-\tau) \mbox{d} \tau. 
	\label{Fsymdot}
	\end{eqnarray}
	Substituting~\eqref{Fsymdot} into~\eqref{Smean1}, and using \eqref{hom3}, gives
	\begin{eqnarray} \nonumber
	\dot{S}_{\text{mes}(n)}(t)&=& \frac{ S_{\text{mes}(n)}(t)}{F_{\text{sym}(n)}(t)} \Bigg[ \frac{n-1}{n} \int_0^t n f_n(\tau) \frac{\dot{S}_{\text{mes}(n)}(t- \tau)}{F_{\text{sym}(n)}(t-\tau)} \mbox{d} \tau \\
	&&-  n f_n(t)(1-y-z)  \Bigg].
	\label{Smean}
	\end{eqnarray}
	It can be shown, using~\eqref{hom11} and~\eqref{Smean} that, for all $ t \in [0,T]$,
	\begin{equation}
	\label{gronwall}
	\left|\dot{S}_{\text{mes}(n)}(t)- \dot{S}(t)\right| \le A(n,T) \int_0^t \left|\dot{S}_{\text{mes}(n)}(u)-\dot{S}(u)\right| \mbox{d} u + B(n,T),
	\end{equation}
	where $B(n,T) \to 0$ as $n \to \infty$ and $0 \le A(n,T) \le 4M_T$ for all sufficiently large $n$ (see appendix \ref{proving}). Application of Gronwall's inequality (see appendix \ref{cont}) then yields that, for all $ t \in [0,T]$,
	\begin{equation} \label{gronout}
	\left|\dot{S}_{\text{mes}(n)}(t)- \dot{S}(t)\right| \le B(n,T) {\rm e}^{A(n,T)t}.
	\end{equation}
	Thus
	\begin{equation*}
	\lim_{n \to \infty} \sup_{0 \le t \le T} \left|\dot{S}_{\text{mes}(n)}(t)- \dot{S}(t)\right|=0,
	\end{equation*} 
	whence
	\begin{eqnarray*}
		\lim_{n \to \infty} \sup_{0 \le t \le T} \left|{S}_{\text{mes}(n)}(t)- {S}(t)\right|
		&=& 
		\lim_{n \to \infty} \sup_{0 \le t \le T} \left| \int_0^t  \dot{S}_{\text{mes}(n)}(u)- \dot{S}(u) \mbox{d} u\right|\\
		&\le& 
		\lim_{n \to \infty} \sup_{0 \le t \le T} \int_0^t \left|\dot{S}_{\text{mes}(n)}(u)- \dot{S}(u)\right|\mbox{d} u\\
		&\le&
		\lim_{n \to \infty} T \sup_{0 \le t \le T} \left|\dot{S}_{\text{mes}(n)}(t)- \dot{S}(t)\right|\\
		=0,
	\end{eqnarray*}
	proving~\eqref{theorem5S}. Equation~\eqref{theorem5R} now follows using a similar argument and~\eqref{theorem5I} is then immediate. \qed
	
\end{proof}

It is straightforward that if $f^*(\tau)= \beta k {\rm e}^{- \gamma \tau}$ and $r(\tau)=\gamma {\rm e}^{- \gamma \tau}$ then the Kermack-McKendrick model reduces to a system of ODEs:
\begin{eqnarray} \label{homS}
\dot {S}(t)&=&-\beta k   S(t) I(t),
\\  \label{homI}
\dot {I}(t)&=&\beta k   S(t) I(t) - \gamma I(t),  \\
\dot {R}(t)&=& \gamma I(t).
\label{hommarkov}
\end{eqnarray}
For this special case, we state the following corollary to Theorem~\ref{kermthm}, see also Wilkinson et al. (2016), where it is proved that, for non-random initial conditions, the Kermack-McKendrick model bounds the so-called `general stochastic epidemic'. 

\begin{corollary}[In the Markovian case, message passing and pairwise models are better approximations than the Kermack-McKendrick model]
	\label{kermboundcor}
	Assume that, in the homogeneous stochastic model, contact and recovery processes are independent and Poisson with rates $\beta$ and $\gamma$ respectively. Specifically, $h(\tau)=\beta {\rm e}^{-\beta \tau}$, $r(\tau)=\gamma {\rm e}^{- \gamma \tau}$ and $f(\tau)=\beta {\rm e}^{-(\beta + \gamma)\tau}$. Let $k$ denote the regular degree of the symmetric graph (instead of $n$). For this special case, 
	\begin{equation} \nonumber
	S(t) < S_{\text{mes}}(t) \le  P_{S}(t), \quad  R(t) > R_{\text{mes}}(t) \ge  P_{R}(t) \qquad (t > 0),
	\end{equation}
	where $S(t)$ and $R(t)$ are given by \eqref{homS}-\eqref{hommarkov}, with $S(0)=z,I(0)=1-y-z$ and $R(0)=y$, and $S_{\text{mes}}(t)$ and $R_{\text{mes}}(t)$ are given by \eqref{otherF}-\eqref{hom} with $n$ replaced by $k$.
\end{corollary}

\begin{proof}
	See appendix \ref{massbound}. \qed
	\end{proof}

\section{Discussion}
\label{disc}
The message passing equations of Karrer and Newman (2010) approximate the expected time course for non-Markovian SIR epidemic dynamics on networks. In a later paper, Wilkinson and Sharkey (2014) slightly generalised their equations in order to make them applicable to stochastic models with more individual level heterogeneity. Here, for the first time, we have shown that Karrer and Newman's system of message passing equations, and its generalisation, have unique feasible solutions (Theorem~\ref{th1}).

An important feature of the message passing equations is that they produce an upper bound to the expected epidemic size (cumulative number of infection events) at every point in time. Thus, they give a `worst case scenario'. In addition, they exactly capture the expected epidemic when the contact network is a tree. Here, we extended these results to a further generalised stochastic model which includes realistic correlations between post-infection contact times and the infectious period (Theorem~\ref{ineqtheorem}). This situation can occur when individuals may adopt disease-combating behaviour, such as taking antiviral medication, which acts both on the ability of an individual to pass on the infection as well as the duration of their infectivity.

Much of this paper was devoted to a special case of the stochastic model which we referred to as the `homogeneous stochastic model', in which individuals are homogeneous and the contact network is a symmetric graph (correlations between post-infection contact times and the infectious period are still allowed). Examples of a symmetric graph include a finite complete graph, an infinite square lattice and an infinite Bethe lattice. Due to symmetry, the message passing system here reduces to just four equations which we refer to as the `homogeneous message passing system'. This system is equivalent in form to a special case of the system found by Karrer and Newman (2010) to describe epidemic dynamics on random configuration networks, but here it is applied to a different stochastic model. These equations were analysed, making use of Theorem~\ref{ineqtheorem}, to obtain a result which shows that cycles in the contact network serve to inhibit the stochastic epidemic (Theorem~\ref{theorem3}). Following arguments from Karrer and Newman (2010), we also obtained a single equation which provides an upper bound on the final epidemic size (Theorem~\ref{outbreakthm}); for the Bethe lattice, the final epidemic size is captured exactly. This naturally provides sufficient conditions, in terms of an $R_0$-like quantity and the level of vaccination, for there to be no major outbreak (Theorem~\ref{outbreakthm}).  

We found that the `limit' of an appropriate sequence of homogeneous message passing systems gives the same epidemic time course as the Kermack-McKendrick (1927) epidemic model (Theorem~\ref{kermthm}) showing that it can be viewed as a special case of message passing. This also has the advantage of relating it to the underlying stochastic model (see also Barbour and Reinert (2013) who establish an exact correspondence). The final epidemic size result, and sufficient conditions for no major outbreak, described above for the homogeneous message passing system then translate directly.

From the homogeneous message passing system, we also constructed an equivalent population-level pairwise system which incorporates a general infectious period (Theorem \ref{pairthm}). This can also be derived directly as a special case of the general individual-level pairwise system of Wilkinson and Sharkey (2014, equations 8 and 9) by applying the conditions of the homogeneous stochastic model. Here we filled a gap in the arguments of this paper by demonstrating sufficient conditions for the valid application of Leibniz's integral rule (appendix~\ref{cont}). This population-level pairwise system contains the Poisson pairwise model (without clustering) of Keeling (1999) as a special case. It also contains the delay differential equation model of Kiss et al. (2015) as a special case. We note that an entirely different derivation of \eqref{mspair}-\eqref{msipair} has been found independently and in parallel by R{\"o}st et al. (2016).

In general, we have emphasised the equivalence between several different types of SIR epidemic model. Specifically, we mention the derivation of the Kermack-McKendrick model as a special case of message passing and the equivalence (under Markovian transmission) of message passing and a class of pairwise models (see also Wilkinson and Sharkey (2014)).  We also note the recently submitted paper by Sherborne et al. (2016) which highlights the equivalence of message-passing and edge-based models (Miller et al., 2011), and that there is equivalence between edge-based models and the model of Volz (2008) (proved by Miller (2011)), and between the model of Volz and the binding site model of Leung and Diekmann (2017, Remark 1). While for SIR dynamics, message passing provides quite a general unifying framework, we note that for other dynamics such as SIS, it remains difficult to formulate a similar construction. 

Unification of models is valuable in narrowing the lines of enquiry and simplifying ongoing research. In addition, owing to their different constructions, different types of results have been more forthcoming for some models than for others, and unification can allow results for one model to be automatically transferred to another. For example, here, by unification with message passing, we have been able to show that when contact and recovery processes are independent and Poisson, the Kermack-McKendrick model (which then reduces to a mass action ODE model) provides a rigorous upper bound on the expected epidemic size at time $t>0$ in the homogeneous stochastic model (Corollary~\ref{kermboundcor}). However, the bound is coarser than that provided by the message passing and pairwise systems, so we now know that these are better approximations. This extends the result that, for non-random initial conditions, the Kermack-McKendrick model bounds the so-called `general stochastic epidemic' (Wilkinson et al. 2016). An interesting development would be to show that the Kermack-McKendrick model \eqref{hom11}-\eqref{hom2} bounds the homogeneous stochastic model more generally. We observe that this could be achieved by showing that the message passing system for the stochastic model is the first in a sequence of message passing systems indexed by $n$, which satisfies the conditions for Theorem \ref{kermthm}, and where $S_{\text{mes}(n)}(t)$ is non-increasing with $n$; this is easy to do for Poisson transmission and recovery processes (appendix~\ref{massbound}). Another extension worthy of investigation is to multitype SIR epidemics.

\appendix

\section{Proof of Theorem \ref{th1}}
\label{feasunique}

 Reproducing an argument from Karrer and Newman (2010), here we construct a feasible (bounded between 0 and 1) solution of (\ref{Fgen1}).
Let $F_{(0)}^{i \leftarrow j}(t)=1$ for all $i \in \mathcal{V}, j \in \mathcal{N}_i$ and all $t \ge 0$, and define the following iterative procedure. For $m=1,2,\ldots$, let
\begin{equation} \label{iter}
F_{(m)}^{i \leftarrow j}(t) = 1- \int_0^t f_{ij} (\tau) \Big(1-y_j-z_j \prod_{k \in \mathcal{N}_j \setminus i} F_{(m-1)}^{j \leftarrow k}(t-\tau)  \Big) \mbox{d} \tau.
\end{equation}
 It is easily shown that $1 \ge F_{(m)}^{i \leftarrow j}(t) \ge F_{(m+1)}^{i \leftarrow j}(t)\ge 1- \int_0^t f_{ij}(\tau) \mbox{d} \tau$, for all $i \in \mathcal{V}, j \in \mathcal{N}_i$, $t \ge 0$ and $m=0,1,\dots$, whence $\vec{F}_m(t) \equiv (F_m^{i \leftarrow j}(t): i \in \mathcal{V}, j \in \mathcal{N}_i)$ converges to some $\vec{F}_{\infty}(t)$ as $m \to \infty$, and $\vec{F}_{\infty}(t)$ is a feasible solution of~\eqref{Fgen1}. Moreover, letting $\vec{F}_*(t)$ be any feasible solution of~\eqref{Fgen1}, it can be shown, arguing as in Corduneanu (1991), section 1.3, that 
 \begin{equation} \label{unique} \sup_{i \in \mathcal{V}, j \in \mathcal{N}_i}  | F_*^{i \leftarrow j}(t)- F_m^{i \leftarrow j}(t)  |  \le  \frac{(N_{\text{max}}-1)^m(t  \,  f_{\text{max}} )^{m+1}}{(m+1)!} , \end{equation}
 where $N_{\text{max}}=\sup_{i \in \mathcal{V}}|\mathcal{N}_i|$ and $f_{\text{max}}=\sup_{i \in \mathcal{V}, j \in \mathcal{N}_i} \sup_{t' \ge 0}f_{ij}(t')$. Assume that $N_{\text{max}}< \infty$ and $f_{\text{max}}< \infty$. Then, the right-hand side of \eqref{unique} converges to zero as $m \to \infty$, and $\vec{F}_{\infty}(t)$ must be the unique feasible solution of~\eqref{Fgen1}.

Note that \eqref{iter} implies that if, for all $i \in \mathcal{V}, j \in \mathcal{N}_i$, it is the case that $F_{(m-1)}^{i \leftarrow j}(t)$ is non-increasing and belongs to $[0,1]$ for all $t \ge 0$, then these properties are also held by $F_{(m)}^{i \leftarrow j}(t) $ for all $i \in \mathcal{V}, j \in \mathcal{N}_i$. Since these properties are held by $F_{(0)}^{i \leftarrow j}(t) (=1)$ for all $i \in \mathcal{V}, j \in \mathcal{N}_i$, then, by induction, they hold for all $m \ge 0$, so $F_{(\infty)}^{i \leftarrow j}(t)$ is non-increasing for all $i \in \mathcal{V}, j \in \mathcal{N}_i$. Thus, the feasible solution of \eqref{Fgraph} (for $F_{\text{sym}}(t)$) is non-increasing, whence $S_{\text{mes}}(t)$ is non-increasing.

To show continuity of the feasible solution, first note that \eqref{iter} implies that if, for all $i \in \mathcal{V}, j \in \mathcal{N}_i$, it is the case that $F_{(m-1)}^{i \leftarrow j}(t)$ is continuous, then $F_{(m)}^{i \leftarrow j}(t) $ is also continuous for all $i \in \mathcal{V}, j \in \mathcal{N}_i$. Since $F_{(0)}^{i \leftarrow j}(t) (=1)$ is continuous for all $i \in \mathcal{V}, j \in \mathcal{N}_i$, then, by induction, $F_{(m)}^{i \leftarrow j}(t) $ is continuous for all $m \ge 0, i \in \mathcal{V}, j \in \mathcal{N}_i$.  Now, for any fixed $T>0$, the bound in~\eqref{unique} holds for all $t \in [0, T]$ provided $t$ in the right-hand side of~\eqref{unique} is replaced by $T$.  Thus  $\vec{F}_m(t)$ converges uniformly to $\vec{F}_{\infty}(t)$ over $[0,T]$ as $n \to \infty$ and, since each $\vec{F}_m(t)$ is continuous on $[0, T]$, it follows that $\vec{F}_{\infty}(t)$ is also continuous on $[0, T]$. 
This holds for any $T>0$, so $\vec{F}_{\infty}(t)$ is continuous on $[0,\infty)$.

\section{Proof of Theorem \ref{boundthm}}
\label{boundappend}

We suppose first that the vertex set $\mathcal{V}$ is finite.  Similarly to Wilkinson and Sharkey (2014, section III), and Ball et al. (2015), it is straightforward to show that the indicator variable $\mathbbm{1}_{i \leftarrow \mathcal{A} (t)}$ for the event that a cavity state-individual $i \in \mathcal{V}$ does not receive any infectious contacts from any of $\mathcal{A} \subset \mathcal{N}_i$ by time $t \ge 0$ is a function of the random variables $\mathcal{X}^{**} \equiv\cup_{i \in \mathcal{V}}\{\mathcal{X}_i^*, Y_i \}$ (see the beginning of section~\ref{second}), and that it is non-decreasing with respect to each element of $\mathcal{X}^{**}$. Thus, since $\mathcal{X}^{**}$ is a set of associated variables (by assumption, and Esary et al. (1967, $(\text{P}_2)$ and $(\text{P}_3)$))) and $Y_i$ is independent of all other members of $\mathcal{X}^{**}$, then using Esary et al. (1967, Theorem 4.1), we have
\begin{equation} \label{inp2} P_{S_i}(t) = z_i {\rm E}[\mathbbm{1}_{i \leftarrow \mathcal{N}_i (t)} ]  \ge z_i \prod_{j \in \mathcal{N}_i} {\rm E}[\mathbbm{1}_{i \leftarrow j (t)} ] = z_i \prod_{j \in \mathcal{N}_i} H^{i \leftarrow j}(t) \qquad (i \in \mathcal{V}), \end{equation}
with equality occurring when the graph is a tree or forest (where putting an individual into the cavity state prevents any dependencies between the states of its neighbours). Recall that $z_i \equiv {\rm P}(Y_i=2)$ is the probability that $i$ is initially susceptible. 

Similarly, the indicator variable $\mathbbm{1}_{(i)j \leftarrow \mathcal{A} (t)}$ for the event that a cavity state-individual $j \in \mathcal{V}$ does not receive any infectious contacts from any of $\mathcal{A} \subset \mathcal{N}_j \setminus i$ by time $t \ge 0$, where $i \in \mathcal{N}_j$ is also in the cavity state, is a function of the random variables $\mathcal{X}^{**}$, and it is non-decreasing with respect to each. Again, since $\mathcal{X}^{**}$ is a set of associated variables then we have (c.f. \eqref{firstH} and \eqref{Fgen1}),
\begin{eqnarray} \nonumber   \Phi_i^j(t) =  {\rm E}[\mathbbm{1}_{(i)j \leftarrow \mathcal{N}_j \setminus i (t)}] & \ge &  \prod_{k \in \mathcal{N}_j \setminus i}{\rm E} [\mathbbm{1}_{(i)j \leftarrow k (t)}] \\ \nonumber
&\ge&  \prod_{k \in \mathcal{N}_j \setminus i} {\rm E}[\mathbbm{1}_{j \leftarrow k (t)}]  \\
&=& \prod_{k \in \mathcal{N}_j \setminus i}H^{j \leftarrow k}(t),\label{inp}  \end{eqnarray}
where the second inequality follows from the fact that taking an individual out of the cavity state cannot increase the probability that a different individual receives no infectious contacts from a given neighbour by time $t \ge 0$. Again, equality occurs when the graph is a tree or forest.

The above derivations of~\eqref{inp2} and~\eqref{inp} break down when the vertex set $\mathcal{V}$  is countably infinite, since the theory in
Esary et al. (1967) requires that the set of random variables $\mathcal{X}^{**}$ is finite.  Suppose now that  $\mathcal{V}$ is countably infinite and
label the vertices $1,2,\dots$.  Fix $i \in \mathcal{V}$ and an integer $n \ge i$.  Let $G^{(n)}=(\mathcal{V}^{(n)}, \mathcal{E}^{(n)})$ be the graph
obtained from $G$ by deleting the vertices $n+1,n+2, \dots$ and all edges connected to those vertices.  Now, since  $|\mathcal{V}^{(n)}| < \infty$, the inequality~\eqref{inp2} yields 
\begin{equation}
\label{inp2finite}
P_{S_i}^{(n)}(t) \ge z_i \prod_{j \in \mathcal{N}^{(n)}_i} H^{(n),i \leftarrow j}(t),
\end{equation}
where the superfix $n$ denotes that the quantity is defined for the epidemic on $G^{(n)}$.  Further, for $n=i,i+1,\dots$, the epidemic on $G^{(n)}$ can be defined using the same set $\mathcal{X}^{**} \equiv\cup_{i \in \mathcal{V}}\{\mathcal{X}_i^*, Y_i \}$ of random variables. It then follows that, for any $t \ge 0$, the event that individual $i$ is susceptible at time $t$ in the epidemic on $G^{(n)}$ decreases with $n$ and tends to the event that individual $i$ is susceptible at time $t$ in the epidemic on $G$ as $n \to \infty$, so $P_{S_i}^{(n)}(t) \to P_{S_i}(t)$ as $n \to \infty$ by the continuity of probability measures.
  A similar argument shows that $H^{(n),i \leftarrow j}(t) \to   H^{i \leftarrow j}(t)$ as $n \to \infty$.
Letting $n \to \infty$ in~\eqref{inp2finite} then shows that~\eqref{inp2} holds when $\mathcal{V}$ is countably infinite, as  $|\mathcal{N}_i| < \infty$.  The same method of proof shows that~\eqref{inp} also holds  when $\mathcal{V}$ is countably infinite.

Using \eqref{inp} in conjunction with \eqref{firstH} we have
\begin{equation} \label{finalHineq}
H^{i \leftarrow j}(t) \ge 1- \int_0^t f_{ij}(\tau) \big(1-y_j-z_j  \prod_{k \in \mathcal{N}_j \setminus i}H^{j \leftarrow k}(t-\tau)  \big) \mbox{d} \tau,
\end{equation}
where equality occurs when the graph is a tree or forest. Using \eqref{finalHineq}, it is easy to show by the iterative procedure in appendix A (except with $F^{i \leftarrow j}_{(0)}(t)=H^{i \leftarrow j}(t)$) that a unique feasible solution of \eqref{Fgen1} exists and, using this solution, that $F^{i \leftarrow j}(t) \le H^{i \leftarrow j}(t)$ for all $i \in \mathcal{V}, j \in \mathcal{N}_i$ and all $t \ge 0$, with equality occurring when the graph is a tree or forest. This fact, in combination with \eqref{inp2}, c.f. \eqref{hom1}, proves \eqref{ineq1}, and consequently, c.f. \eqref{Rm}, gives \eqref{ineq2}.

\section{Proof of Theorem \ref{pairthm}}
\label{pairappendix}
Here we consider the homogeneous stochastic model defined at the beginning of section \ref{symsec} with reference to the beginning of section \ref{second}. We assume that transmission processes are Poisson with rate $\beta$ and that they are independent of the recovery processes, specifically $f(\tau)=\beta {\rm e}^{- \beta \tau} \int_{\tau}^{\infty} r(\tau') \mbox{d}\tau'$. We assume that $r(\tau)$ is continuous so that we may apply Leibniz's integral rule to compute derivatives (see appendix \ref{cont}). In this case, a pairwise system incorporating a general infectious period can be derived from the homogeneous message passing system (\ref{hom3})-(\ref{Fgraph}) with the additional variables:
\begin{eqnarray} \label{pair1} SS_{\text{mes}} (t) & \equiv & z^2 F_{\text{sym}}(t)^{2(n-1)}, \\ \label{pair2} SI_{\text{mes}}(t) & \equiv & z F_{\text{sym}}(t)^{n-1} \Bigg( \frac{-  \dot{F}_{\text{sym}}(t)}{\beta} \Bigg) , \end{eqnarray}
where $SS_{\text{mes}}(t)$ approximates the probability that a pair of neighbours are susceptible at time $t$, and $SI_{\text{mes}}(t)$ approximates the probability that the first is susceptible and the second is infected at time $t$ (see Wilkinson and Sharkey (2014, section II B) where these pairwise quantities were first considered in the context of message passing). To understand the construction of the factor in brackets in (\ref{pair2}), note that for any pair of neighbours $i,j,$ the probability that $i$ is susceptible and $j$ is infected at time $t$ remains the same when $i$ is placed into the cavity state. Further, when transmission processes are Poisson with rate $\beta$, we must have that:
\begin{eqnarray} \nonumber \label{factor} \dot{H}^{i \leftarrow j}(t) &=& - \beta \, {\rm P}( \text{$j$ infected at time $t$ and no infectious} \\
 && \qquad \text{contacts from $j$ to $i$ before time $t$} \mid \text{$i$ in cavity}) .  \end{eqnarray}
Thus, the factor in brackets in (\ref{pair2}) can be seen to approximate the probability on the right-hand side of (\ref{factor}) for any pair of neighbours $i,j$ (recall that $F_{\text{sym}}(t)$ approximates $H^{i \leftarrow j}(t)$ for any pair of neighbours $i,j$).

 To obtain population-level quantities, we define (as in Sharkey (2008, appendix B)):
\begin{equation} \label{pop}
[S](t) \equiv N S_{\text{mes}}(t), \quad [I](t) \equiv N I_{\text{mes}}(t), \quad [SS](t) \equiv n N SS_{\text{mes}}(t), \quad [SI](t) \equiv n N SI_{\text{mes}}(t),
\end{equation}
where $N$ is a positive number. Note that~\eqref{hom3} and~\eqref{pair2} imply
\begin{equation}
\label{pairF}
\dot{F}_{\text{sym}}(t) = - \beta F_{\text{sym}}(t) \frac{ SI_{\text{mes}}(t)}{ S_{\text{mes}}(t)} \qquad (S_{\text{mes}}(t) \neq 0),
\end{equation}
so, since $F_{\text{sym}}(0)=1$, we have:
\begin{equation} \label{pairF2}
F_{\text{sym}}(t) = \text{exp} \Bigg( - \int_0^t \beta  \frac{ SI_{\text{mes}}(\tau)}{ S_{\text{mes}}(\tau)} \mbox{d} \tau \Bigg) \qquad (S_{\text{mes}}(t)  \neq 0).
\end{equation}

 Substituting from (\ref{hom3})-(\ref{Rpair}) and (\ref{pair1}), and using (\ref{pairF}), it is straightforward to write down the time derivatives of $[S](t)$, $[I](t)$ and $[SS](t)$ as in (\ref{mspair})-(\ref{msspair}). 

Finding the time derivative of $[SI](t)$ is more involved. Setting $u=t-\tau$ in~\eqref{Fgraph} and differentiating with respect to $t$ using Leibniz's integral rule yields, recalling $f(\tau)= \beta {\rm e}^{- \beta \tau} \int_{\tau}^{\infty} r(\tau') \mbox{d} \tau'$, that
\begin{eqnarray} \nonumber
\dot{F}_{\text{sym}}(t)&=& - \beta \big( F_{\text{sym}}(t) -y -z F_{\text{sym}}(t)^{n-1} \big) \\ \label{Fpairderiv}
&& + \int_0^t \beta {\rm e}^{- \beta \tau} r (\tau) \Big(  1-y-z F_{\text{sym}}(t - \tau)^{n-1} \Big) \mbox{d} \tau.
\end{eqnarray}
Substituting from (\ref{pair2}) and (\ref{Fpairderiv}) into (\ref{pop}), we can write
\begin{eqnarray} \nonumber
[SI](t)&=& n N z F_{\text{sym}}(t)^{n-1} \Bigg( \frac{- \dot{F}_{\text{sym}}(t)}{\beta} \Bigg) \\ \nonumber
&=&  n N z F_{\text{sym}}(t)^{n-1} \Bigg[ F_{\text{sym}}(t) - y - z F_{\text{sym}}(t)^{n-1} \\ \label{SIeq}
 &&- \int_0^{t}  {\rm e}^{- \beta \tau} r (\tau) \Big( 1 - y - z F_{\text{sym}}(t-\tau)^{n-1} \Big) \mbox{d} \tau      \Bigg].
\end{eqnarray}
Differentiating the right-hand side of (\ref{SIeq}), we can now express the time derivative of $[SI](t)$ as
\begin{eqnarray} \label{SIde} \nonumber
\dot{[SI]}(t)&=& n (n-1) N z F_{\text{sym}}(t)^{n-2} \dot{F}_{\text{sym}}(t) \Bigg( \frac{- \dot{F}_{\text{sym}}(t)}{\beta} \Bigg) \\ \nonumber
&& + n N z F_{\text{sym}}(t)^{n-1} \dot{F}_{\text{sym}}(t) \\ \nonumber
&& -  n (n-1) N z^2 F_{\text{sym}}(t)^{2n-3} \dot{F}_{\text{sym}}(t) \\ \nonumber
&& + n (n-1) N z^2 F_{\text{sym}}(t)^{n-1}  \int_0^t {\rm e}^{- \beta \tau} r (\tau) F_{\text{sym}}(t-\tau)^{n-2} \dot{F}_{\text{sym}}(t- \tau) \mbox{d} \tau \\
&& -  n N z F_{\text{sym}}(t)^{n-1}{\rm e}^{- \beta t} r (t) (1-y-z).
\end{eqnarray}
Substituting from (\ref{hom3}),(\ref{pair1}),(\ref{pair2}),(\ref{pop}),(\ref{pairF}) and (\ref{pairF2}) into (\ref{SIde}) yields the expression for $\dot{[SI]}(t)$ in (\ref{msipair}); the terms on the right-hand side of (\ref{msipair}) are ordered by equality with the terms on the right-hand side of (\ref{SIde}).

\section{Continuity conditions for the application of Leibniz's integral rule and Gronwall's inequality}
\label{cont}

To derive \eqref{Fsymdot}, Leibniz's integral rule is applied to \eqref{Fgraph}, and this is valid if $F_{\text{sym}}(t)$ is continuously differentiable. Similarly, the application of the rule in the derivation of \eqref{Fpairderiv} and \eqref{SIde} is valid if $f(\tau)$ and $F_{\text{sym}}(t)$ are continuously differentiable. Here we show that $F_{\text{sym}}(t)$ is continuously differentiable if $f(\tau)$ is continuously differentiable. Note that if $f(\tau)= \beta e^{- \beta \tau} \int_{\tau}^{\infty} r(\tau') \mbox{d} \tau'$ then $f(\tau)$ is continuously differentiable when $r(\tau)$ is continuous.

With reference to the message passing system, \eqref{hom3}-\eqref{Fgraph}, assume that $f(\tau)$ is continuously differentiable. Thus we may apply Leibniz's integral rule to \eqref{Fgraph}, after setting $\tau'=t-\tau$, in order to compute the derivative of $F_{\text{sym}}(t)$ as follows
\begin{equation} \label{Fdotcont} \dot{F}_{\text{sym}}(t)= - \int_0^t \dot{f}(t- \tau')(1 - y - z F_{\text{sym}}(\tau')^{n-1}) \mbox{d}\tau' - f(0)(1 - y - z F_{\text{sym}}(t)^{n-1}). \end{equation}
It follows from Appendix~\ref{feasunique} that $F_{\text{sym}}(t)$ is continuous.  Thus, since $\dot{f}(\tau)$ is also continuous, \eqref{Fdotcont} implies that $\dot{F}_{\text{sym}}(t)$ is continuous. 

To derive \eqref{gronout}, Gronwall's inequality is applied to \eqref{gronwall}, and this is valid if $\dot{S}_{\text{mes}(n)}(t)$ and $\dot{S}(t)$ are continuous. By condition (iii) of Theorem \ref{kermthm}, we have that $\dot{F}_{\text{sym}(n)}(t)$ is continuous (by the above argument), so $\dot{S}_{\text{mes}(n)}(t)$ is continuous. Conditions (i) and (iii) imply that $f^*(t)$ is continuous, which implies that $\dot{S}(t)$ is continuous.

We note that Leibniz's integral rule was assumed to be applicable in Wilkinson and Sharkey (2014). It is straightforward, using a similar argument to above, to show that the application of the rule in that paper is valid if $f_{ij}(\tau)$ is continuously differentiable for all $i \in V, j \in N_i$. 

\section{Proof of \eqref{gronwall}}
\label{proving}

It follows from~\eqref{hom11} and~\eqref{Smean} that, for all $t \in [0,T]$,
\begin{equation}
\label{modSSn}
\left|\dot{S}_{\text{mes}(n)}(t)- \dot{S}(t)\right| \le A_n(t)+B_n(t),
\end{equation}
where
\begin{equation*}
A_n(t)=\left|\frac{ S_{\text{mes}(n)}(t)}{F_{\text{sym}(n)}(t)} \Bigg[ \frac{n-1}{n} \int_0^t n f_n(\tau) \frac{\dot{S}_{\text{mes}(n)}(t- \tau)}{F_{\text{sym}(n)}(t-\tau)} \mbox{d} \tau  \Bigg]-S(t) \int_0^t  f^*(\tau) \dot{S}(t- \tau) \mbox{d} \tau \right|
\end{equation*}
and
\begin{equation*}
B_n(t)=\left|\frac{ S_{\text{mes}(n)}(t)}{F_{\text{sym}(n)}(t)} n f_n(t)(1-y-z)-S(t)I(0)f^*(t)\right|.
\end{equation*}

Now 
\begin{equation}
\label{AnAn1An2}
A_n(t)\le A_n^{(1)}(t)+A_n^{(2)}(t),
\end{equation}
where
\begin{align*}
A_n^{(1)}(t)=\left|\frac{ S_{\text{mes}(n)}(t)}{F_{\text{sym}(n)}(t)}\right. \Bigg[ \frac{n-1}{n} \int_0^t n f_n(\tau)& \frac{\dot{S}_{\text{mes}(n)}(t- \tau)}{F_{\text{sym}(n)}(t-\tau)} \mbox{d} \tau  \Bigg]\\
&  -\left.S_{\text{mes}(n)}(t) \int_0^t  f^*(\tau) \dot{S}(t- \tau) \mbox{d} \tau \right|
\end{align*}
and
\begin{equation*}
A_n^{(2)}(t)=\left|S_{\text{mes}(n)}(t)-S(t)\right|\times\left| \int_0^t f^*(\tau) \dot{S}(t- \tau) \mbox{d} \tau \right|.
\end{equation*}
Considering $ A_n^{(1)}(t)$, note that, since $0 \le S_{\text{mes}(n)}(t) \le 1$, 
\begin{equation}
\label{An1bound}
A_n^{(1)}(t) \le \left(\frac{n-1}{n}\right)\frac{1}{F_{\text{sym}(n)}(t)}A_n^{(11)}(t)+ A_n^{(12)}(t),
\end{equation} 
where
\begin{eqnarray*}
	A_n^{(11)}(t) &=&\left|\int_0^t  n f_n(\tau) \frac{\dot{S}_{\text{mes}(n)}(t- \tau)}{F_{\text{sym}(n)}(t-\tau)} \mbox{d} \tau - \int_0^t  f^*(\tau) \dot{S}(t- \tau) \mbox{d} \tau\right|\\
	&\le&\left|\int_0^t\frac{n f_n(\tau)}{F_{\text{sym}(n)}(t-\tau)}\left(\dot{S}_{\text{mes}(n)}(t- \tau)-\dot{S}(t- \tau)\right)\mbox{d} \tau\right|\\
	&& \quad +\left|\int_0^t \left(\frac{n f_n(\tau)}{F_{\text{sym}(n)}(t-\tau)}-f^*(\tau)\right)\dot{S}(t- \tau)\mbox{d} \tau\right|
\end{eqnarray*}
and
\begin{equation*}
A_n^{(12)}(t)=\left|\int_0^t  f^*(\tau) \dot{S}(t- \tau) \mbox{d} \tau \right|\times\left|\left(\frac{n-1}{n}\right)\frac{1}{F_{\text{sym}(n)}(t)}-1\right|.
\end{equation*}
Now conditions (i), (ii) and~\eqref{Flargen} imply that, for all $t \in [0,T],\tau \in [0,t]$, 
\begin{equation*}
\frac{n f_n(\tau)}{F_{\text{sym}(n)}(t-\tau)} \le \frac{M_T+\epsilon_n(T)}{1-\epsilon_n^{(1)}(T)}
\end{equation*}
and
\begin{eqnarray} \nonumber
\left|\frac{n f_n(\tau)}{F_{\text{sym}(n)}(t-\tau)}-f^*(\tau)\right|& \le &\frac{1}{F_{\text{sym}(n)}(t-\tau)}\Big( \left|n f_n(\tau)-f^*(\tau)\right|\\
&&+f^*(\tau)\left(1-F_{\text{sym}(n)}(t-\tau)\right)\Big) \nonumber \\
&\le& \frac{\epsilon_n(T)+M_T \epsilon_n^{(1)}(T)}{1-\epsilon_n^{(1)}(T)},\label{nfFbound}
\end{eqnarray}
whence
\begin{align*}
A_n^{(11)}(t) &\le\frac{M_T+\epsilon_n(T)}{1-\epsilon_n^{(1)}(T)} \int_0^t \left|\dot{S}_{\text{mes}(n)}(t- \tau)- \dot{S}(t- \tau)\right| \mbox{d} \tau\\
&\qquad+ \frac{\epsilon_n(T)+M_T \epsilon_n^{(1)}(T)}{1-\epsilon_n^{(1)}(T)} \left|\int_0^t \dot{S}(t- \tau) \mbox{d} \tau \right|\\
&\le
\frac{M_T+\epsilon_n(T)}{1-\epsilon_n^{(1)}(T)} \int_0^t \left|\dot{S}_{\text{mes}(n)}(u)- \dot{S}(u)\right| \mbox{d} u +\frac{\epsilon_n(T)+M_T \epsilon_n^{(1)}(T)}{1-\epsilon_n^{(1)}(T)},
\end{align*}
as $\int_0^t \dot{S}(t- \tau) \mbox{d} \tau = S(0)-S(t) \in [0,1]$.  A similar argument, noting that
\begin{equation*}
\left|\int_0^t  f^*(\tau) \dot{S}(t- \tau) \mbox{d} \tau \right|\le \int_0^t \left|f^*(\tau) \dot{S}(t- \tau)\right| \mbox{d} \tau
\le M_T[S(0)-S(t)]\le M_T,
\end{equation*}
shows that
\begin{equation*}
A_n^{(12)}(t) \le \frac{M_T\left(\epsilon_n^{(1)}(T)+\frac{1}{n}\right)}{1-\epsilon_n^{(1)}(T)}.
\end{equation*}
Hence, recalling~\eqref{An1bound},
\begin{eqnarray} \nonumber
A_n^{(1)}(t) &\le& \frac{M_T+\epsilon_n(T)}{\left(1-\epsilon_n^{(1)}(T)\right)^2}\int_0^t \left|\dot{S}_{\text{mes}(n)}(u)- \dot{S}(u)\right| \mbox{d} u \\
&&+ \frac{\epsilon_n(T)+M_T \epsilon_n^{(1)}(T)}{\left(1-\epsilon_n^{(1)}(T)\right)^2} +\frac{M_T\left(\epsilon_n^{(1)}(T)+\frac{1}{n}\right)}{1-\epsilon_n^{(1)}(T)}. \label{An1fin}
\end{eqnarray}

Turning to $A_n^{(2)}(t)$, note that since $S_{\text{mes}(n)}(0)=S(0)$,
\begin{eqnarray*}
	\left|S_{\text{mes}(n)}(t)-S(t)\right|&=&\left| \int_0^t \dot{S}_{\text{mes}(n)}(u)-\dot{S}(u) \mbox{d} u\right|\\
	&\le&
	\int_0^t \left|\dot{S}_{\text{mes}(n)}(u)-\dot{S}(u)\right| \mbox{d} u,
\end{eqnarray*}
so
\begin{equation} \label{An2fin}
A_n^{(2)}(t) \le M_T \int_0^t \left|\dot{S}_{\text{mes}(n)}(u)-\dot{S}(u)\right| \mbox{d} u.
\end{equation}

Further, since $I(0)=1-y-z$ and $0 \le I(0), S_{\text{mes}(n)}(t) \le 1$,
\begin{eqnarray} \nonumber
	B_n(t)&=&I(0)\left|\frac{ S_{\text{mes}(n)}(t)}{F_{\text{sym}(n)}(t)} n f_n(t)-S(t)f^*(t)\right|\\ \nonumber
	&\le&I(0)\left(f^*(t)\left|S_{\text{mes}(n)}(t)-S(t)\right|+ S_{\text{mes}(n)}(t)\left|\frac{n f_n(t)}{F_{\text{sym}(n)}(t)}-f^*(t)\right|\right)\\ \label{Bnfin}
	&\le& M_T \int_0^t \left|\dot{S}_{\text{mes}(n)}(u)-\dot{S}(u)\right| \mbox{d} u
	+\frac{\epsilon_n(T)+M_T \epsilon_n^{(1)}(T)}{1-\epsilon_n^{(1)}(T)},
\end{eqnarray}
using a similar result to~\eqref{nfFbound}. 

Thus, using \eqref{modSSn}, \eqref{AnAn1An2}, \eqref{An1fin}, \eqref{An2fin} and \eqref{Bnfin}, we may define
\begin{equation*}
	A(n,T)=2M_T+\frac{M_T+\epsilon_n(T)}{\left(1-\epsilon_n^{(1)}(T)\right)^2}
	\end{equation*}
	and
	\begin{equation*}
	B(n,T)=\frac{\left(\epsilon_n(T)+M_T \epsilon_n^{(1)}(T)\right)\left(2-\epsilon_n^{(1)}(T)\right)}{\left(1-\epsilon_n^{(1)}(T)\right)^2}+\frac{M_T\left(\epsilon_n^{(1)}(T)+\frac{1}{n}\right)}{1-\epsilon_n^{(1)}(T)},
	\end{equation*}
such that inequality \eqref{gronwall} is satisfied for all $t \in [0,T]$. Further, since both $\epsilon_n(T)$ and $\epsilon_n^{(1)}(T)$ converge to $0$ as $n \to \infty$, it follows that $B(n,T) \to 0$ as $n \to \infty$ and $0 \le A(n,T) \le 4 M_T$ for all sufficiently large $n$.

\section{Proof of Corollary \ref{kermboundcor}}
\label{massbound}
Here, we consider the homogeneous stochastic model (defined at the beginning of section \ref{symsec}, with reference to the beginning of section \ref{second}) for the special case where transmission and recovery processes are independent and Poisson with rates $\beta$ and $\gamma$ respectively. Specifically, $h(\tau)=\beta {\rm e}^{-\beta \tau}$, $r(\tau)=\gamma {\rm e}^{- \gamma \tau}$ and $f(\tau)=\beta {\rm e}^{-(\beta + \gamma)\tau}$. For convenience, we let $k$ denote the regular degree of the symmetric graph (instead of $n$).

For this special case, we show here that for the same initial conditions and parameters,
\begin{equation}  \label{kermboundapp}
P_{S}(t) \ge S_{\text{mes}}(t) > S(t) \qquad  \mbox{for all  } t >0,
\end{equation}
where $P_{S}(t)$ is the probability that an arbitrary individual is susceptible at time $t$ (this being the same for all individuals) and $S(t)$ is given by the special case of the Kermack-McKendrick model (\ref{homS})-(\ref{hommarkov}), with $S(0)=z>0, I(0)=1-y-z>0$ and $R(0)=y$; $S_{\text{mes}}(t)$ is given by (\ref{hom3}) and (\ref{Fgraph}) but with $n$ replaced by $k$. Note that since
\begin{eqnarray} \nonumber
P_R(t)&=& y + \int_0^t \gamma {\rm e}^{- \gamma \tau}\Big(1-y-P_S(t-\tau) \Big) \mbox{d} \tau, \\
\nonumber
R_{\text{mes}}(t)&=& y + \int_0^t \gamma {\rm e}^{- \gamma \tau}\Big(1-y-S_{\text{mes}}(t-\tau) \Big) \mbox{d} \tau,
\end{eqnarray}
and
\begin{eqnarray} \nonumber
R(t)&=& y + \int_0^t \gamma {\rm e}^{- \gamma \tau}\Big(1-y-S(t-\tau) \Big) \mbox{d} \tau,
\end{eqnarray}
then \eqref{kermboundapp} implies that $R(t) > R_{\text{mes}}(t) \ge P_{R}(t)$ for all $t>0$.

We already have $P_{S}(t)\ge S_{\text{mes}}(t)$ by Theorem~\ref{boundthm} and the fact that the message passing system, in this case, has a unique solution. Thus, we may prove \eqref{kermboundapp} and Corollary~\ref{kermboundcor} by showing that $S_{\text{mes}}(t) > S(t)$ for all $t>0$.

Setting $f_n(\tau)= (\beta k/n){\rm e}^{-(\beta k/n+ \gamma) \tau }$ and $f^*(\tau)=\beta k {\rm e}^{- \gamma \tau }$, the Kermack-McKendrick model reduces to the system of ODEs (\ref{homS})-(\ref{hommarkov}) and the conditions for Theorem~\ref{kermthm} are satisfied. Thus, letting $F_{\text{sym}(n)}(t)$ be defined by (\ref{Fgraph}) but with $f(\tau)$ replaced by $f_{n}(\tau)$, and letting $S_{\text{mes}(n)}(t)$ be defined by (\ref{hom3}) but with $F_{\text{sym}}(t)$ replaced by $F_{\text{sym}(n)}(t)$ (as in subsection~\ref{large}),
$$\lim_{n \to \infty} S_{\text{mes}(n)}(t)=S(t)$$
and
$$S_{\text{mes}(n)}(t)=S_{\text{mes}}(t) \qquad \mbox{if $n=k$}.$$
Therefore, if $S_{\text{mes}(n)}(t) \equiv z F_{\text{sym}(n)}(t)^n$ is strictly decreasing with respect to $n$, for all $t > 0$, then we have $S_{\text{mes}}(t) > S(t)$ for all $t > 0$. We now show this to be the case.

Letting $u_n(t)=F_{\text{sym}(n)}(t)^n (=S_{\text{mes}(n)}(t)/z)$, we can write (c.f. \eqref{case1a})
$$  \dot{u}_{n}(t) = n \gamma \Big( u_{n}(t)^{\frac{n-1}{n}} - u_{n}(t) \Big) - \beta k  \Big( u_{n}(t)-yu_{n}(t)^{\frac{n - 1}{n}}- z u_{n}(t)^{\frac{2(n - 1)}{n}}   \Big).$$
For fixed $u \in (0,1)$, we have that $u^{\frac{n -1}{n}}$ is strictly decreasing with $n$, and also that
\begin{eqnarray} \nonumber
n(  u^{\frac{n-1}{n}}-u ) &=& n u (     u^{\frac{-1}{n}}-1  ) \\ \nonumber
&=& n {\rm e}^{- \lambda} ({\rm e}^{\frac{\lambda}{n}}-1) \qquad (\mbox{where }u={\rm e}^{- \lambda}, \mbox{ so } \lambda > 0) \\ \nonumber
&=& {\rm e}^{- \lambda} \sum_{k=1}^{\infty} \frac{1}{k !} \frac{\lambda^k}{n^{k-1}}
\end{eqnarray}
is strictly decreasing with $n$. Therefore, since $u_n(0)=1$ and $u_n(t) \in (0,1)$ for $t>0$, it follows that $u_{n}(t)$ (and hence $S_{\text{mes}(n)}(t)$) is strictly decreasing with $n$ for all $t>0$. 

\section*{Acknowledgements}
R.R.W. acknowledges support from EPSRC (DTA studentship). R.R.W. and K.J.S. acknowledge support from the Leverhulme Trust (RPG-2014-341). We thank the reviewers and associate editor for their constructive comments which have improved the presentation of the paper.



\begin{thebibliography}{}
%
%


  \bibitem{Anderson} Anderson RM, May RM (1992) Infectious diseases of humans. Oxford University Press.

\bibitem{Bailey} Bailey NTJ (1975) The Mathematical Theory of Infectious Diseases. Griffin, London.


 \bibitem{Ball3} Ball FG, Wilkinson RR, Sharkey KJ (2015) Erratum: Message passing and moment closure for susceptible-infected-recovered epidemics on finite networks. Phys. Rev. E 92, 039904.
 
\bibitem{Barbour72} Barbour AD (1972) The principle of the diffusion of arbitrary constants. J. Appl. Probab. 9(3), 519--541.
 
\bibitem{Barbour74} Barbour AD (1974) On a functional central limit theorem for Markov population processes. Adv. Appl. Probab. 6(1), 21--39.

\bibitem{Barbour} Barbour AD, Reinert G (2013) Approximating the epidemic curve. Electron. J. Probab. 18, no. 54, 1--30.

\bibitem{Corduneanu} Corduneanu C (1991) Integral Equations and Applications. Cambridge
University Press.

\bibitem{Danon} Danon L, Ford AP, House T, Jewell CP, Keeling MJ, Roberts GO, Ross JV, Vernon MC (2011) Networks and the epidemiology of infectious disease. Interdisciplinary Perspectives on Infectious Diseases. doi:10.1155/2011/284909.

\bibitem{Diekmann} Diekmann O, de Jong MCM, Metz JAJ (1998) A deterministic epidemic model taking account of repeated contacts between the same individuals. J. Appl. Probab. 35(2), 448--462.

 \bibitem{Don} Donnelly P (1993) The correlation structure of epidemic models. Math. Biosci. 117(1-2), 49--75.

 \bibitem{Esary} Esary JD, Proschan F, Walkup DW (1967) Association of random variables, with applications. Ann. Math. Stat. 38(5), 1466--1474.



\bibitem{Godsil} Godsil C, Royle G (2001) Algebraic Graph Theory. New York. Springer.



 \bibitem{Karrer} Karrer B, Newman MEJ (2010) A message passing approach for general epidemic models. Phys. Rev. E. 82, 016101.

\bibitem{Keeling} Keeling MJ (1999) The effects of local spatial structure on epidemiological invasions. Proc. Biol. Sci., 266, 589--867.

  \bibitem{Kermack}  Kermack WO, McKendrick AG (1927) A contribution to the mathematical theory of epidemics. Proc R Soc Lond A 115, 700--721.

\bibitem{Kiss} Kiss IZ, R\"{o}st G, Zsolt V (2015) Generalization of pairwise models to non-Markovian epidemics on networks. Phys. Rev. Lett. 115, 078701.


\bibitem{Kurtz70} Kurtz TG (1970) Solutions of ordinary differential equations as limits of pure jump Markov processes. J. Appl. Probab. 7(1), 49--58.

\bibitem{Kurtz71} Kurtz TG (1971) Limit theorems for sequences of jump Markov processes approximating ordinary differential processes. J. Appl. Probab. 8(2), 344--356.


\bibitem{Leung17} Leung KY, Diekmann O (2017) Dangerous connections: on binding site models of infectious disease dynamics.  J. Math. Biol. 74, 619--671.

 \bibitem{Miller} Miller JC (2009) The spread of infectious disease through clustered populations. J. R. Soc. Interface, doi:10.1098/rsif.2008.0524.
 
  \bibitem{miller4}   Miller JC (2011) A note on a paper by Erik Volz: SIR dynamics in random networks. J. Math. Biol. 62(3), 349--358.
  
   \bibitem{Miller2} Miller JC (2012) A note on the derivation of epidemic final sizes. Bull. Math. Biol. 74(9), 2125--2141.

  \bibitem{Miller2} Miller JC, Slim AC, Volz EM (2011) Edge-based compartmental modelling
  for infectious disease spread. J. R. Soc. Interface 9, 890--906.




\bibitem{Pastor} Pastor-Satorras R, Castellano C, Van Mieghem P, Vespignani A (2015) Epidemic processes in complex networks. Rev. Mod. Phys. 87, 925.

\bibitem{Rost} R\"{o}st G, Vizi Z, Kiss IZ (2016) Pairwise approximation for SIR type network epidemics with non-Markovian recovery. arXiv:1605.02933.


\bibitem{Sharkey} Sharkey KJ (2008) Deterministic epidemiological models at the individual level. J. Math. Biol. 57, 311--331.

\bibitem{Sherborne} Sherborne N, Miller JC, Blyuss KB, Kiss IZ (2016) Mean-field models for non-Markovian epidemics on networks: from edge-based compartmental to pairwise models. arXiv:1611.04030.

\bibitem{Simon} Simon PL, Taylor M, Kiss IZ (2011) Exact epidemic models on graphs using graph-automorphism driven lumping. J. Math. Biol. 62(4), 479--508.


\bibitem{Trapman} Trapman P (2007) Reproduction numbers for epidemics on networks using pair approximation. Math. Biosci. 210, 464--489.

\bibitem{Volz} Volz E (2008) SIR dynamics in random networks with heterogeneous connectivity. J. Math. Biol. 56(3), 293--310. 




\bibitem{Wilkinson} Wilkinson RR, Ball FG, Sharkey KJ (2016) The deterministic Kermack-McKendrick model bounds the general stochastic epidemic. Scheduled for publication in J. Appl. Probab. 53(4):1031-1040.

\bibitem{Wilkinson} Wilkinson RR, Sharkey KJ (2014) Message passing and moment closure for susceptible-infected-recovered epidemics on finite networks. Phys. Rev. E 89, 022808.





%
\end{thebibliography}
\end{document}